\documentclass[aps,pra,10pt,a4paper,reprint,nofootinbib,,superscriptaddress]{revtex4-2}
\usepackage{xcolor}
\usepackage[utf8]{inputenc}
\usepackage{ulem}
\usepackage{amsfonts}
\usepackage{amssymb}
\usepackage{amsmath}
\usepackage{amsthm}
\usepackage{graphics}
\usepackage{graphicx}
\usepackage{braket}
\usepackage[colorlinks=true,linkcolor=blue,urlcolor=blue,citecolor=blue]{hyperref}
\usepackage{times,txfonts}
\usepackage{orcidlink}

\newtheorem{definition}{Definition}
\newtheorem{proposition}{Proposition}
\newtheorem{remark}{Remark}

\usepackage[T1]{fontenc}

\newcommand{\suptiny}[3]{\ensuremath{^{\hspace{#1 pt}\protect\raisebox{#2 pt}{\tiny{$ #3$}}}}}

\begin{document}
\title{Protocols for sharing genuine multipartite entanglement by employing copies of biseparable states}

\author{Swati Choudhary\orcidlink{0009-0005-4542-4290}}

\affiliation{Center for Quantum Science and Technology and  
Center for Computational Natural Sciences and Bioinformatics,
International Institute of Information Technology Hyderabad, Prof. CR Rao Road, Gachibowli, Hyderabad 500 032, Telangana, India}
\affiliation{Harish-Chandra Research Institute, Chhatnag Road, Jhunsi, Prayagraj  211 019, India}
\author{Ujjwal Sen\orcidlink{0000-0002-0091-5847}}
\affiliation{Harish-Chandra Research Institute,  Chhatnag Road, Jhunsi, Prayagraj  211 019, India}
\affiliation{Homi Bhabha National Institute, Training School Complex, Anushakti Nagar, Mumbai
400 094, India}

\author{Saronath Halder\orcidlink{0000-0001-5738-5245}}
\affiliation{Department of Physics, School of Advanced Sciences, VIT-AP University, Beside AP Secretariat, Amaravati 522 241, Andhra Pradesh, India}

\begin{abstract}
Sharing genuine multipartite entanglement by considering collective use of copies of biseparable states, which are entangled across all bipartitions but lack genuine multipartite entanglement at the single-copy level, plays a central role in several quantum information processing protocols, and has been referred as  genuine multipartite entanglement activation. We present a protocol for three-qutrit systems showing that two copies of rank-two biseparable states, entangled across every bipartition, are sufficient to generate a genuinely multipartite entangled state with nonzero probability. This contrasts with the three-qubit scenario where many copies of biseparable states might be required for sharing genuine multipartite entanglement. We subsequently generalize our protocols to the case of an arbitrary number of parties. Interestingly, the proposed construction naturally leads to the activation of genuinely nonlocal correlations, yielding a result that is stronger than genuine multipartite entanglement activation alone.
\end{abstract}

\maketitle
\section{Introduction}\label{Intro}
The activation~\cite{WLP2013,ADNPS2014,KLMG2012,P2012,QSetal2018,TKB2020,BH2021,V2023,YMMBFH2022,PV2022,BLMAF2025} of quantum resources~\cite{SV2009, CH2016, CG2019, TR2019, SR2020, MM2024, BR2024} has emerged as a fundamental theme in quantum information science, driven by the fact that essential resources such as entanglement~\cite{HHHH2009} and nonlocality~\cite{BCPSW2014} are highly fragile to noise and decoherence~\cite{BOOK1996,Z2003,S2005,IK2022}. Understanding how to recover or enhance these resources when they are initially absent or insufficient is crucial for the development of robust quantum technologies. From the perspective of both foundational studies and practical implementations, protocols that convert weaker or differently structured correlations into operationally useful resources play a dual role. For example, they expand our conceptual understanding of entanglement as a resource. Furthermore, they also relax stringent experimental requirements imposed by imperfect state preparation and environmental noise.  A particularly challenging instance of this general task is the activation of genuine multipartite entanglement (GME)~\cite{YMMBFH2022,PV2022,TLHVT2024,LKSSXNO2024}. Informally, the question that one usually address in the context of GME activation is the following. We consider two or more copies of a multipartite state that is not genuinely multipartite entangled but that nevertheless displays entanglement across every bipartition (i.e., the state is biseparable but entangled in all bipartitions). Is it possible that the parties employing only local operations and classical communication (LOCC), generate a state that is genuinely multipartite entangled from the copies of the given state? Resolving this question is important because there are tasks for which it is necessary to use genuine multipartite correlations~\cite{HB1999,CGL1999,T2010,EKS2016,C2017,PHS2021} and the presence of GME is prerequisite for a broad family of genuine multipartite correlations~\cite{GLTSvL2016,KLHP2021,KLP2022,ACMG2007}.

Recent research has focused on the activation of genuine multipartite entanglement (GME), where the initial states are not genuinely multipartite entangled but belong to the class of biseparable states. Notably, works such as~\cite{YMMBFH2022} have investigated activation in special families of states, for example, isotropic GHZ states which are highly symmetric states.   Most activation protocols assume access to at least two copies of the underlying biseparable state and rely on correlated local measurements across these copies to produce genuinely entangled outputs.  While these constructions are powerful in theory, the reliance on symmetry (for example, studies that focus on highly symmetric isotropic GHZ families \cite{YMMBFH2022}) limit applicability in realistic settings. This useful but limited knowledge of GME activation motivates the search for activation protocols that place milder assumptions on available resources. In particular, a protocol that operates on a single copy at a time, does not presuppose symmetry of the state, and avoids auxiliary genuinely entangled resources, would be  conceptually informative. At present, however, a complete characterization of which biseparable states admit such activation and by what class of operations remains an open problem.

In this work, we introduce conceptually simple \textcolor{black}{sequential extraction-and-fusion} schemes for GME activation that departs from simultaneous multi-copy requirement in all steps. Our protocol first distill~\cite{BBPSSW1996,BDSW1996,D1996, B1996,HHH1997,HHH1998} out entangled bipartite pairs between the parties considering copies of the biseparable states adaptively. These extracted singlets are then consumed to produce a multipartite genuinely entangled state shared among all parties with some nonzero probability.  Therefore, we can say that the principal advantages of our approach are two-fold. First, 
our protocol does not require all parties to possess their respective shares of every copy at the outset. Rather, it proceeds sequentially, employing copies of biseparable states in a stepwise manner. In particular, bipartite entanglement is first established between selected pairs of parties (for example, in a four-partite scenario, between \(AB\), \(BC\), and \(CD\)). Subsequently, intermediate parties (such as \(B\) and \(C\)) perform unsharp measurements on their subsystems, thereby enabling the generation of a genuinely multipartite entangled state. Thus, at the extraction stage, simultaneous utilization of all copies is not required. However, at the fusion stage, the implementation of joint measurement(s) appears to be essential for achieving genuine multipartite entanglement. \textcolor{black}{Note that the operational feature of the present constructions rises from the separation between the extraction of bipartite entangled links and their subsequent fusion. This may not be interpreted as a general reduction in experimental or hardware complexity. We just consider two steps: copies of biseparable states are processed separately during link extraction, while joint local measurements on the extracted subsystems are employed at the fusion stage. An alternative approach} may involve a teleportation-based protocol, wherein a local genuinely multipartite entangled state is first prepared and its subsystems are subsequently distributed sequentially using the extracted states. Unlike teleportation based protocol, our proposed protocol does not require any of the party to locally prepare a genuinely entangled state. Furthermore, we also try to reduce the number of copies to achieve the GME activation with some nonzero probability. Second, in every protocol, we basically create pure multipartite genuinely entangled states. Therefore, our protocols distills pure state entanglement from mixed states and since, they are pure, they may exhibit genuine multipartite nonlocality\footnote{Genuine multipartite nonlocality is a type of multipartite correlation which is stronger than genuine multipartite entanglement. This is in the following sense: for genuine multipartite nonlocal correlation, genuine multipartite entanglement is necessary, however, all genuine multipartite entangled states do not exhibit genuine multipartite nonlocal correlation \cite{BCPSW2014}.} which is stronger than just GME activation. In this way, our result generalizes that of \cite{HSTT2003}, which proves that all rank-two bipartite states are distillable. Here, we extend this to the multipartite scenario by showing that rank-two tripartite biseparable states that are entangled across all bipartitions are useful to distill GME states. 

The remainder of this paper is organized as follows. In Section~\ref{Preliminaries}, we review the necessary preliminaries that will be used in subsequent discussions. In Section~\ref{results}, we present results that include new protocols for GME activation. These protocols inherently account for the activation of genuine nonlocality as well. Finally, we conclude our discussion in Section~\ref{Conclusion}.

\section{Preliminaries}\label{Preliminaries}
In this section, we outline the theoretical framework and summarize the relevant literature that constitute the foundation of our analysis.

\begin{definition}
\textit{Pure k-separable state}~\cite{HHHH2009,GT2009,DCLSSS2016,GCS2023}. A pure quantum state of an $n$-partite system with Hilbert space $\mathcal{H}\suptiny{0}{0}{(n)}=\bigotimes_{i=1}^{n}\mathcal{H}_{i}$ is called to be as k-separable, if it can be written as a product of pure states of a maximum of k sub-systems, viz., $\ket{\Psi_{k-sep}}=\ket{\Psi_{1}}\otimes\ket{\Psi_{2}}\otimes....\ket{\Psi_{k}}$, where k$\le$n.
\end{definition}

\begin{definition}
\textit{Mixed k-separable state}~\cite{GHG2010}.  A density operator is called $k$-separable if it can be decomposed as a convex sum of pure states that are all separable with respect to some $k$-partition, i.e., if it is of the form :
\begin{align}
\rho\suptiny{0}{0}{(k)} &=\,
\sum\limits_{i} p_{i} 
\ket{\phi_i\suptiny{0}{0}{(k)}}\!\!\bra{\phi_i\suptiny{0}{0}{(k)}}
\,.\label{ksep}
\end{align}
Each $\ket{\phi_i\suptiny{0}{0}{(k)}}$ may be $k$-separable with respect to~a different $k$-partition. $k=2$ corresponds to biseparable states. 
\end{definition}

\begin{definition}
All pure states that are not at least biseparable are called pure genuinely $n$-partite entangled states. Formally, these states correspond to $k=1$.
\end{definition}

{\it Our task:} We want to consider multiple copies of biseparable states. These states are mixed states and they are entangled across every bipartition. The copies of the states are given one by one, i.e., one copy is given at a time. In this situation, our task is to produce a pure genuinely entangled state with some nonzero probability. Ultimately, multiple parties share genuine multipartite correlation. Formally, we say the following about GME activation.


\begin{definition}
Genuine multipartite entanglement (GME) activation~\cite{YMMBFH2022,PV2022,TLHVT2024,LKSSXNO2024}: Creating GME state(s) with some nonzero probability from multiple copies of multipartite biseparable states that are entangled across all bipartitions, by means of local operations is termed as GME activation.
\end{definition}
From the definition of mixed biseparable states, it follows that a biseparable state need not be separable with respect to a fixed bipartition. Instead, such states may appear as convex mixtures of states separable across different bipartitions. Hence, while they lack genuine multipartite entanglement, they may still exhibit bipartite entanglement in distinct partitions. As an illustration, consider the three-qubit biseparable state of the following form;
\begin{align}
\rho
= p\ket{\phi^{+}}\bra{\phi^{+}}\otimes \ket{0}\bra{0} 
+ (1-p) \ket{1}\bra{1} \otimes \ket{\phi^{+}}\bra{\phi^{+}},
\end{align}
where $\ket{\phi^{+}} = (\ket{00} + \ket{11})/\sqrt{2}$ is the maximally entangled Bell state and $0 < p < 1$. Here, with probability $p$, subsystems $A$ and $B$ share $\ket{\phi^+}$ while $C$ is in the pure state $\ket{0}$; with probability $(1-p)$, entanglement resides between $B$ and $C$ while $A$ is in $\ket{1}$. Thus, $\rho_{ABC}$ is biseparable without being separable across any fixed bipartition. 

We are now ready to present our main findings.

\section{Results}\label{results}
We start with a technical result for three qubits.

\begin{proposition}\label{prop1}
Any three-qubit rank-2 biseparable state that is entangled across every bipartition, is useful for genuine multipartite entanglement (GME) activation.  
\end{proposition}

\begin{proof}
To prove the proposition, we first have to understand the form of a three-qubit rank-2 biseparable state that is entangled across every bipartition. Then, considering many copies of such a state, one should achieve genuine entanglement. 

We consider that the state is shared among three parties (A)lice, (B)ob, and (C)harlie. The form of the state is given as the following. 
\begin{equation}\label{3-qubit}
\rho = p\ket{\Phi}\bra{\Phi}\otimes\ket{\phi}\bra{\phi}+(1-p)\ket{\psi}\bra{\psi}\otimes\ket{\Psi}\bra{\Psi},
\end{equation}
where $\ket{\Phi}\bra{\Phi}\otimes\ket{\phi}\bra{\phi}$ and $\ket{\psi}\bra{\psi}\otimes\ket{\Psi}\bra{\Psi}$ are two rank-1 biseparable states, convex combination of which produces a rank-2 bisebarable state, $p$ is nonzero probability. $\ket{\Phi}$ and $\ket{\Psi}$ are two-qubit entangled states and they are shared between different pairs of parties. These pairs of parties are \{A, B\} and \{B, C\}. The states $\ket{\phi}$ and $\ket{\psi}$ are single qubit pure states hold by C and A respectively. It is important to note that the pure states considered in the decomposition of $\rho$ are product in different bipartitions. This crucial fact enables the mixed state $\rho$ to exhibit entanglement across all bipartitions. Alternatively, if $\rho$ is entangled across every bipartition, then the states $\ket{\Phi}\bra{\Phi}\otimes\ket{\phi}\bra{\phi}$ and $\ket{\psi}\bra{\psi}\otimes\ket{\Psi}\bra{\Psi}$ must be separable in different bipartitions. Here, it is AB-C and A-BC bipartitions respectively. We mention that the type of states we talk about here can have different forms other than the above form. For instance, one may consider convex combinations of states that are product with respect to other bipartitions, such as $B{:}AC$, $A{:}BC$, or $C{:}AB$. An entirely analogous argument then applies in the sense that the resulting mixed state remains entangled across all bipartitions, while the constituent pure states are product with respect to the corresponding alternative bipartitions. Although the details of the protocol may be modified accordingly, its overall structure remains unchanged. Therefore, it suffices to focus on the form given in Eq.~(\ref{3-qubit}). In particular, if GME activation can be achieved for the state in this form, the same conclusion extends to states of these alternative forms as well.

In the next step, we assume that the party C performs a projective measurement on the state $\rho$. The measurement is defined by the operators, $\{\ket{\phi}\bra{\phi}, \mathbb{I}-\ket{\phi}\bra{\phi}\}$, where $\mathbb{I}$ is the identity operator acting on the qubit Hilbert space. Since, the state $\ket{\Psi}$ is a two-qubit entangled state, it must have the form $\ket{\phi_1}\ket{\phi}+\ket{\phi_1^\prime}\ket{\phi^\prime}$ (not normalized), where the pairs of states $\{\ket{\phi_1}, \ket{\phi_1^\prime}\}$, $\{\ket{\phi}, \ket{\phi^\prime}\}$ form two-qubit bases. Clearly, in the measurement of C if the operator $\ket{\phi}\bra{\phi}$ clicks then the post-measurement state becomes another rank-2 state which has the following form. 
\begin{equation}
\rho^\prime = p^\prime\ket{\Phi}\bra{\Phi}\otimes\ket{\phi}\bra{\phi}+(1-p^\prime)\ket{\psi}\bra{\psi}\otimes\ket{\alpha}\bra{\alpha}\otimes\ket{\phi}\bra{\phi},
\end{equation}
where $\ket{\alpha}$ is a single qubit pure state, hold by the party B, $p^\prime$ is nonzero probability. So, if we trace out the subsystem of C then we are left with a rank-2 bipartite two-qubit entangled state. The proof of entanglement follows from the fact that it is a convex combination of a pure entangled state and a product state \cite{HSTT2003}. Such a state is known to be distillable \cite{HHH1997, HSTT2003}, i.e., if sufficient copies of the state are available, then, one can have two-qubit maximally entangled state. In this way, one can produce a maximally entangled state between A and B starting from many copies of $\rho$ with some nonzero probability. Using other copies of $\rho$, if the protocol starts with a projective measurement by A defined via the operators $\{\ket{\psi}\bra{\psi}, \mathbb{I}-\ket{\psi}\bra{\psi}\}$, then it is possible to produce a maximally entangled state between B and C with some nonzero probability following the similar steps. 

In this way, from many copies of $\rho$ it is possible to produce a perfect resource state with some nonzero probability and the rest is due to a teleportation based protocol, i.e., B can produce any three-qubit GME state locally and then, B can teleport one qubit to A and another qubit to C. In this way, the activation of GME state occurs.

If they do not want to proceed with a teleportation based protocol, they can do the following. Having two two-qubit pure entangled states, one is shared between Alice-Bob and the other is shared between Bob-Charlie, these bipartite entangled states can be combined to activate genuine entanglement among all three parties A, B, and C. The joint state is
\begin{equation}\label{step2}
\begin{array}{l}
\ket{\psi} = \ket{\psi}_{AB} \otimes \ket{\psi}_{BC}\\[1 ex]
= a^2 \ket{0}\ket{00}\ket{0} + b^2 \ket{1}\ket{11}\ket{1} \\[1 ex]
+ ab \big(\ket{0}\ket{01}\ket{1} + \ket{1}\ket{10}\ket{0}\big),
\end{array}
\end{equation}
where $\ket{\psi}_{AB}=a\ket{00}+b\ket{11}$ and $\ket{\psi}_{BC}=a\ket{00}+b\ket{11}$ written in the Schmidt form, hence $a,b$ are real with $a^2+b^2=1$ and $a=b=\frac{1}{\sqrt{2}}$ corresponds to maximally entangled state. Bob performs a joint measurement on the two qubits he is holding, defined by- $\Big\{P_1,P_2\Big\}:=\Big\{ \ket{00}\bra{00} + \ket{11}\bra{11},\; \ket{01}\bra{01} + \ket{10}\bra{10}\Big\}$. If $P_1$ clicks, the output is proportional to $a^2 \ket{0}\ket{00}\ket{0}  + b^2 \ket{1}\ket{11}\ket{1}$. Bob then measures in the $\{ \lvert + \rangle, \lvert - \rangle \}$ basis on, say, the first qubit he is holding. This gives $\ket{+} \otimes \left( a^{2}\ket{000} + b^{2}\ket{111} \right)$. He then traces out the first qubit. So that the final state that $ABC$ share becomes $a^{2}\ket{000} + b^{2}\ket{111} 
\quad \text{(up to normalization).}$ Note that corresponding to $P_{2}$, final state would be the similar type.
\end{proof}

To illustrate the above proposition with an example, consider a tripartite setting where three spatially separated parties Alice, Bob, and Charlie share a three-qubit biseparable state, defined by-
\begin{align}
\rho
= p\ket{\phi^{+}}\bra{\phi^{+}}\otimes \ket{0}\bra{0} 
+ (1-p) \ket{1}\bra{1} \otimes \ket{\phi^{-}}\bra{\phi^{-}},
\end{align}
where $\ket{\phi^+}$ and $\ket{\phi^-}$ are two-qubit entangled states [$\ket{\phi^\pm}=(1/\sqrt{2})$ $(\ket{00}\pm\ket{11})$] and they are shared between different pairs of parties. These pairs of parties are \{A, B\} and \{B, C\}. The states $\ket{0}$ and $\ket{1}$ are possessed by C and A respectively. We analyze the effect of a local projective measurement performed by Charlie on his qubit in the computational basis $\{\ket{0}, \ket{1}\}$. This measurement yields two possible outcomes. With nonzero probability, the outcome $\ket{0}$ is obtained, in which case the unnormalized post-measurement state of Alice and Bob (after tracing out the qubit of Charlie) is given by-
\begin{align}
\tilde{\rho} = \ket{\phi^+}\bra{\phi^+} + \ket{10}\bra{10}
\end{align}
It is straightforward to verify that $\tilde{\rho}$ is inseparable as it is a convex combination of a pure entangled state and a product state \cite{HSTT2003}. In fact, the present state is a two-qubit entangled state which must have distillable entanglement \cite{HHH1997}. Thus, $\tilde{\rho}$ allows the extraction of a maximally entangled Bell state between A and B from sufficiently many copies of the initial biseparable state with some nonzero probability. It is important to note that the post-measurement state corresponding to $\ket{1}$ is not considered in our analysis. This is because the resulting state in this case is fully separable, implying the absence of any entanglement. Since, separable states cannot provide any advantage in the context of our objective, they are irrelevant for the subsequent discussion. However, analogous to the construction of the post-measurement state $\tilde{\rho}$, we can obtain a similar state $\tilde{\rho}^\prime$ when Alice performs a measurement in the computational basis on another copy of biseparable state given in the above. The resulting state $\tilde{\rho}^\prime$ also has distillable entanglement and it can be transformed into a maximally entangled Bell state between parties B and C with some nonzero probability. The final step in our construction is the activation of genuine multipartite entanglement (GME) using these two Bell states. This activation process is straightforward and can be done by a teleportation based protocol. Bob can prepare a three-qubit GME state. Then, he can distribute one qubit to Alice and another qubit to Charlie. We now present the following remark.

\begin{remark}\label{rem1}
Any biseparable state is unable to produce genuine nonlocality. However, after our protocol, we are left with a pure three-qubit Greenberger–Horne–Zeilinger (GHZ) state (with some nonzero probability) which can produce genuine nonlocality. Therefore, one can say that many copies of rank-2 biseparable states which are entangled across every bipartition, are also useful to produce genuine nonlocality. Consequently, the proposed protocol is effective not only for the activation of genuine entanglement but also for the activation of genuine nonlocality.
\end{remark}

However, a major drawback of the protocol, discussed so far, is that we are using a distillation protocol in GME activation. This is because for the distillation process, a large number of copies of the initial state are usually required \cite{BBPSSW1996}. So, in the following, we try to construct a protocol for GME activation which does not depend on a distillation process (like the protocol of \cite{BBPSSW1996}). This is basically to reduce the number of copies required for GME activation.

\subsection{Higher dimensional states for GME activation}\label{subsec1}
Here we consider the activation of a GME state in a more general scenario, where the initially shared biseparable state is no longer a three-qubit state. Specifically, we first study the biseparable states $\rho$ which are associated with $\mathcal{C}^3\otimes\mathcal{C}^3\otimes\mathcal{C}^3$.

\begin{proposition}\label{prop2}
There are three-qutrit rank-2 biseparable states which are entangled across every bipartition. From two copies of such a state, it is possible to obtain a genuinely entangled state with some nonzero probability. 
\end{proposition}

\begin{proof}
The three-qutrit state that we consider here is given by-
\begin{equation}\label{3-qutrit}
\rho = p\ket{\psi}\bra{\psi}\otimes\ket{0}\bra{0} + (1-p)\ket{0}\bra{0}\otimes\ket{\psi}\bra{\psi},
\end{equation}
where $\ket{\psi} = \sum_{i=0}^{2}a_{i}\ket{ii}$, written in Schmidt form. $\rho$ is a rank-2 state which is a convex combination of two pure biseparable states $\ket{\psi}\bra{\psi}\otimes\ket{0}\bra{0}$ and $\ket{0}\bra{0}\otimes\ket{\psi}\bra{\psi}$. In $\ket{\psi}\bra{\psi}\otimes\ket{0}\bra{0}$, the bipartite entangled state is shared between A and B. On the other hand, in $\ket{0}\bra{0}\otimes\ket{\psi}\bra{\psi}$, the bipartite entangled state is shared between B and C.

Next, we consider that Charlie performs a projective measurement defined by - \(\{\ket{0}\bra{0}, \mathbb{I}-\ket{0}\bra{0}\}\) on a copy of $\rho$. The resultant  post-measurement state corresponding to the outcome $\mathbb{I}-\ket{0}\bra{0}$, is $\ket{0}\bra{0}\otimes\ket{\psi^\prime}\bra{\psi^\prime}$, where $\ket{\psi^\prime}$ is a two-qubit entangled state. This is because the projector $\mathbb{I}-\ket{\phi}\bra{\phi}$ projects the state $\ket{\psi}$ onto a two-dimensional subspace. Subsequently, tracing out the subsystem of A yields a pure bipartite entangled state between B and C. Similarly, if Alice also does the same projective measurement on a different copy of $\rho$, following the same argument as in Charlie's case, a pure bipartite entangled state can be prepared between A and B. In this way, from two copies of $\rho$, with some nonzero probability, it is possible to produce two copies of $\ket{\psi^\prime}$, one is shared between A and B, then, the other is shared between B and C. Now, we have already discussed that when such a resource is shared among three parties, a genuinely entangled state can be prepared among those parties. See Fig.~\ref{fig:problem} for your reference. Thus, we arrive to the above proposition. 
\end{proof}

\begin{figure}
\centering
\includegraphics[width=\linewidth]{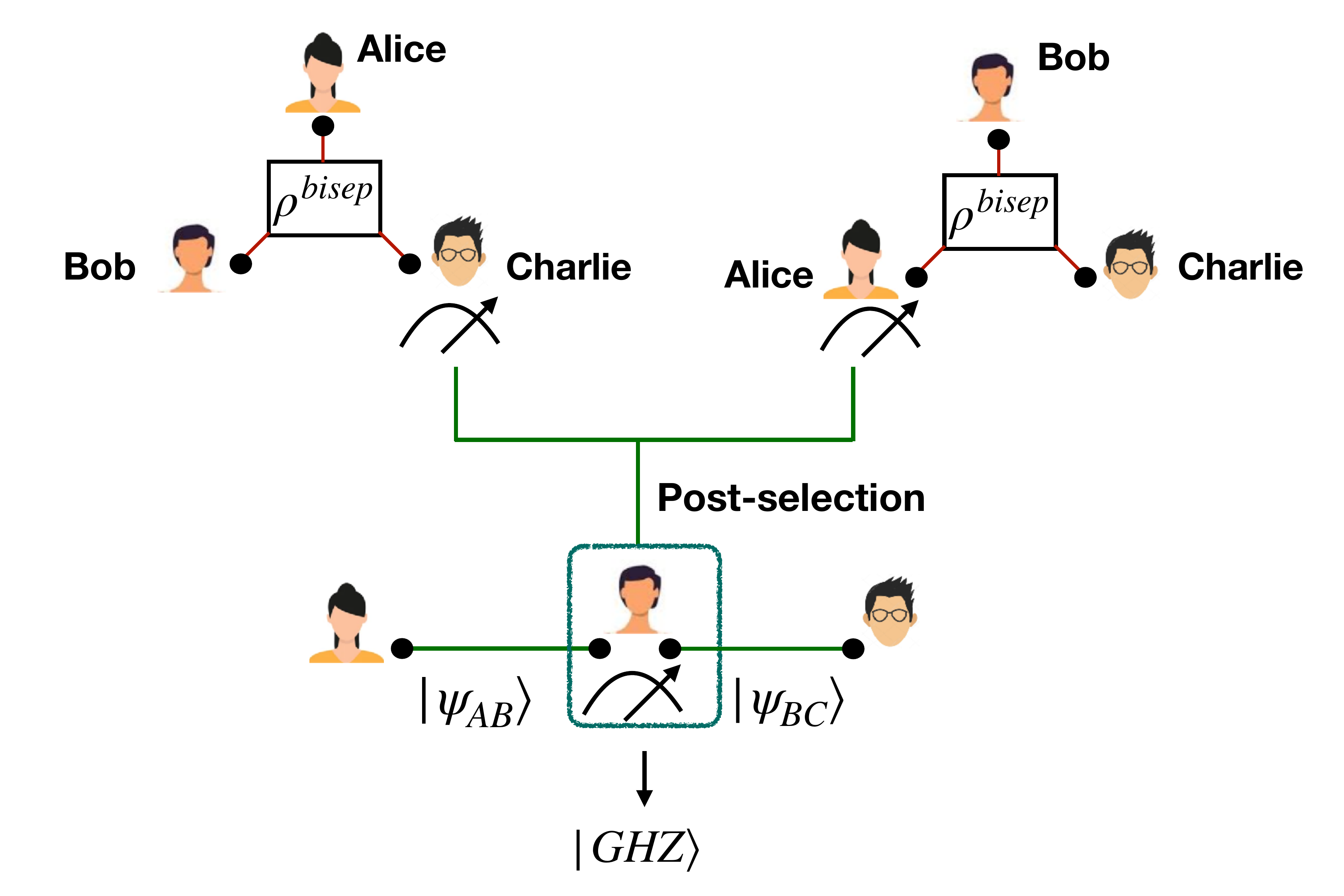}
\caption{Pictorial representation of the protocol for 3-qutrit systems: Suppose A(lice), B(ob) and C(harlie) share among them two copies of biseparable state as shown in Eq.~(\ref{3-qutrit}). Both Charlie and Alice do measurement defined by - \(\{\ket{0}\bra{0}, \mathbb{I}-\ket{0}\bra{0}\}\) on different copies and selects the preferable outcome state, i.e., the one corresponding to the measurement operator \(\{\mathbb{I}-\ket{0}\bra{0}\}\). Subsequently, tracing out certain part of the system gives non-maximally entangled state, in general, on AB and BC, i.e., $\ket{\psi_{AB}}\otimes \ket{\psi_{BC}}$. Then, using calculation from Eq.~(\ref{step2}), a genuinely entangled state can be shared between Alice, Bob and Charlie.}\label{fig:problem}
\end{figure}

The main difference between Proposition \ref{prop1} and Proposition \ref{prop2} is that in the later case we are considering only two copies (minimum number of copies) of the biseparable states for GME activation while in the former case many copies are required for GME activation. This improvement is happening because in the second case the composition states of $\rho$ have higher Schmidt ranks in certain bipartitions. If we choose not to consider higher Schmidt rank, then also it is possible to do GME activation from lower number of copies. Such states are given in the following subsection. Notice that the protocol which is described for Proposition \ref{prop2}, can be easily generalized for qudits where the dimension of a subsystem can be greater than three.

Based on Proposition \ref{prop1} and Proposition \ref{prop2}, it is also possible to classify the biseparable states that are useful to GME activation into two categories: (i) a class of biseparable states for which many copies are necessary to obtain a GME state probabilistically, (ii) the other class of biseparable states for which only two copies (minimum number) are sufficient to obtain a GME state probabilistically. Such a classification is useful to understand the structures of biseparable states in the context of GME activation.

We also mention that for both protocols corresponding to Proposition \ref{prop1} and Proposition \ref{prop2}, the copies of the biseparable states are given one by one. In fact, at present it is not known if there is any advantage of joint local operations on such copies.

\subsection{Success probability for a large number of copies}
Here we consider a specific class of mixed biseparable states which are associated with tripartite Hilbert space \(\mathbb{C}^3 \otimes \mathbb{C}^2 \otimes \mathbb{C}^3\). The state is a convex mixture of two pure biseparable terms and it is entangled across every bipartition. The class of states is given by-  

\begin{align}
\sigma = p \ket{2}\bra{2}\otimes \ket{\phi^+}\!\bra{\phi^+} + (1-p)\ket{\phi^+}\bra{\phi^+}\otimes \ket{2}\bra{2},
\end{align}
where \(\ket{\phi^+} = \frac{1}{\sqrt{2}}(\ket{00} + \ket{11})\) represents the standard Bell state, and \(0<p<1\). Consider that Alice performs a two-outcome projective measurement on her subsystem, which is a qutrit, using the orthogonal projectors $\{ P^A_1, P^A_2 \} := \{ \ket{0}\!\bra{0} + \ket{1}\!\bra{1}, \ket{2}\!\bra{2} \}$. This measurement effectively distinguishes whether Alice’s state lies in the qubit subspace \(\text{span}\{\ket{0},\ket{1}\}\) or in the state \(\ket{2}\). A straightforward calculation shows that if the outcome corresponds to \(P^A_1\), which happens with probability \((1-p)\), then Alice and Bob share the maximally entangled state \(\ket{\phi^+}_{AB}\). Conversely, if the outcome corresponds to \(P^A_2\), which happens with probability \(p\), then Bob and Charlie share \(\ket{\phi^+}_{BC}\). In either case, the post-measurement state between two parties is a Bell pair, and thus a maximally entangled state is obtained with certainty (success probability \(100\%\)). Although, the specific pair of parties which are sharing the entanglement, depends on the measurement outcomes.

Suppose, without loss of generality, that Alice obtains the first outcome and hence Alice and Bob share a Bell pair. We now introduce an additional independent copy of \(\sigma\). On this second copy, Charlie performs an analogous two-outcome measurement \(\{P^C_1, P^C_2\}\), defined by $\{ P^C_1, P^C_2 \} := \{ \ket{0}\!\bra{0} + \ket{1}\!\bra{1}, \ket{2}\!\bra{2} \}$. Conditioned on Charlie’s measurement, the outcomes exhibit similar behavior: with probability \(p\), Bob and Charlie share a Bell pair, while with probability \((1-p)\), Alice and Bob share a Bell pair. Since Alice’s measurement in the first copy already produced a Bell pair between \(A\) and \(B\), Charlie accepts only the outcome where \(B\) and \(C\) share a Bell pair, which occurs with probability \(p\). [This is because the three parties want to produce here a perfect tripartite resource state using which a tripartite genuine entangled state can be distributed.] If the desired outcome does not occur, Charlie discards the copy and repeats the measurement on subsequent copies until success is achieved.

So, the probability of success after \(n\) independent trials is given by  $P_n = 1 - (1 - p)^n,$ which can approach to unity in the asymptotic limit (number of copies is very large):  
\begin{align}
\lim_{n \to \infty} P_n = 1 - \lim_{n \to \infty} (1 - p)^n = 1.
\end{align}
Finally, assume that Bob locally prepares a tripartite entangled state, such as the GHZ state or the W state. Using the two shared Bell pairs, one is shared between \(A\) and \(B\) and the other is shared between \(B\) and \(C\), obtained from the copies of \(\sigma\), Bob can teleport two qubits of the tripartite state to Alice and Charlie. Consequently, this protocol demonstrates that many copies of the biseparable state \(\sigma\), although not genuinely entangled, can act as a resource for GME activation. 

If the three parties share copies of the biseparable state $\sigma^\prime$ = $p\ket{2}\bra{2}\otimes\ket{\phi^\prime}\bra{\phi^\prime}+(1-p)\ket{\phi^\prime}\bra{\phi^\prime}\otimes\ket{2}\bra{2}$, where $\ket{\phi^\prime}$ is a non-maximally entangled two-qubit state, $0<p<1$, then to create a perfect resource state, appropriate for a teleportation based protocol, additional steps of measurements are required. Nevertheless, if the parties choose not to consider these additional steps, then, they can proceed with protocol of creating a GHZ state, described just before Remark \ref{rem1}, instead of considering a teleportation based protocol.  

\subsection{Generalization}
Here we provide possible generalizations for multipartite systems (number of parties $>3$) of the protocols that are described previously for tripartite systems. We start with the following proposition.

\begin{proposition}\label{prop3}
There are rank-3 biseparable states in $\mathbb{C}^4\otimes\mathbb{C}^4\otimes\mathbb{C}^4\otimes\mathbb{C}^4$ which are entangled across every bipartition. From three copies of such a state, it is possible to obtain a genuinely entangled state with some nonzero probability. 
\end{proposition}

\begin{proof}
The state that we consider here is given by-
\begin{equation}
\begin{aligned}
\rho
=\;& p_1\,\ket{\psi}\bra{\psi}_{AB}
\otimes \ket{0}\bra{0}_{C}
\otimes \ket{0}\bra{0}_{D} \\
&+ p_2\,\ket{0}\bra{0}_{A}
\otimes \ket{\psi}\bra{\psi}_{BC}
\otimes \ket{1}\bra{1}_{D} \\
&+ p_3\,\ket{1}\bra{1}_{A}
\otimes \ket{1}\bra{1}_{B}
\otimes \ket{\psi}\bra{\psi}_{CD},
\end{aligned}
\end{equation}
where $\ket{\psi} = \sum_{i=0}^{3}a_{i}\ket{ii}$, written in Schmidt form and $p_{1}+p_{2}+{p_3}=1$ with $0<p_{i}<1$ $\forall$ $ i\in \{1,2,3\}$. $\rho$ is a rank-3 state which is a convex combination of three pure biseparable states. $\rho$ is shared among four parties A, B, C, and D. Let the party, C performs a projective measurement defined by - \(\{\ket{0}\bra{0}, \ket{1}\bra{1}, \mathbb{I}-\ket{0}\bra{0}- \ket{1}\bra{1}\}\) on a copy of $\rho$. The resultant  post-measurement state corresponding to the projector $\mathbb{I}-\ket{0}\bra{0}- \ket{1}\bra{1}$, is then given by $p_2^\prime\,\ket{0}\bra{0}_{A}
\otimes \ket{\psi^\prime}\bra{\psi^\prime}_{BC}
\otimes \ket{1}\bra{1}_{D}+p_3^\prime\ket{1}\bra{1}_{A}
\otimes \ket{1}\bra{1}_{B}
\otimes \ket{\psi^\prime}\bra{\psi^\prime}_{CD}$, where $\ket{\psi^\prime}$ is a two-qubit entangled state because the projector $\mathbb{I}-\ket{0}\bra{0}- \ket{1}\bra{1}$ projects the state $\ket{\psi}$ onto a two-dimensional subspace, $p_2^\prime+p_3^\prime=1$, $0<p_i^\prime<1$, $i=2,3$. If the party D further does the same measurement and get the same outcome like the party C, then, the four parties are left with $\ket{1}\bra{1}_{A}
\otimes \ket{1}\bra{1}_{B}
\otimes \ket{\psi^\prime}\bra{\psi^\prime}_{CD}$. Subsequently, tracing out the subsystems of A and B yields a pure bipartite entangled state between C and D. Similarly, if A and B perform the same projective measurements on a different copy of $\rho$ and obtain outcomes like C and D, following the same arguments as in the case of C and D, a pure bipartite entangled state can be prepared between A and B. Finally, Considering outcomes of the measurements of B and C for another copy of $\rho$, it is possible to prepare a pure bipartite entangled state between B and C. In this way, from three copies of $\rho$, with some nonzero probability, it is possible to produce three copies of $\ket{\psi^\prime}$, one is shared between A and B, one is shared between B and C then, the final one is shared between C and D. The joint state is similar to the state, given by-
\begin{equation}
\begin{array}{l}
\ket{\Psi} = \ket{\phi}_{AB} \otimes \ket{\phi}_{BC} \otimes \ket{\phi}_{CD}\\[1 ex]
= a^3 \ket{0}\ket{00}\ket{00}\ket{0} + b^3 \ket{1}\ket{11}\ket{11}\ket{1} \\[1 ex]
+ a^{2}b \big(\ket{0}\ket{00}\ket{01}\ket{1} + \ket{0}\ket{01}\ket{10}\ket{0}+\ket{1}\ket{10}\ket{00}\ket{0}\big)\\[1 ex]
+ ab^{2}\big(\ket{1}\ket{11}\ket{10}\ket{0}+\ket{1}\ket{10}\ket{01}\ket{1}+ \ket{0}\ket{01}\ket{11}\ket{1}\big),
\end{array}
\end{equation}
where $\ket{\phi}_{AB}=a\ket{00}+b\ket{11}$, $\ket{\phi}_{BC}=a\ket{00}+b\ket{11}$ and $\ket{\phi}_{CD}=a\ket{00}+b\ket{11}$ written in the Schmidt form, hence $a,b$ are real with $a^2+b^2=1$ and $a=b=\frac{1}{\sqrt{2}}$ when it corresponds to a maximally entangled state. Let B performs a joint measurement on the two qubits he is holding, defined by- $\Big\{P_1,P_2\Big\}:=\Big\{ \ket{00}\bra{00} + \ket{11}\bra{11},\; \ket{01}\bra{01} + \ket{10}\bra{10}\Big\}$. If $P_1$ clicks, the output is proportional to $a^3 \ket{0}\ket{00}\ket{00}\ket{0}  +  a^2b\ket{0}\ket{00}\ket{01}\ket{1} + b^2 a \ket{1}\ket{11}\ket{10}\ket{0} +b^3 \ket{1}\ket{11}\ket{11}\ket{1}$. Let C also does the same measurement on 2 qubits he is holding. Then, If $P_1$ clicks given that in Bob's measurement $P_{1}$ had clicked,  the output is proportional to $a^3 \ket{000000}+b^3 \ket{111111}$. B and C then measures in the $\{ \lvert + \rangle, \lvert - \rangle \}$ basis on, say, the first qubits they are holding and then tracing that out gives final state $a^{3}\ket{0000} + b^{3}\ket{1111} 
\quad \text{(up to normalization).}$ Note that corresponding to other three combinations of measurement outcomes, final state will be of the similar type.
\end{proof}

The generalization of Proposition \ref{prop3} for higher number of parties ($\geq5$) is quite straightforward. We also mention that similar argument like Remark \ref{rem1} is applicable here as well. For the sake of completeness, we give the generalization for higher number parties($\ge 5$) as the remark given below. \textcolor{black}{However, there is a limitation to our proposed protocol in the sense that as the number of parties increases, the rank of the biseparable state and the copies of such states required
to do GME activation may increase. Also, experimental complexity of the protocol is not under the scope of this work and hence
has not been discussed in the work and is left for further research.} 

\begin{remark}\label{re}
There are d-partite rank-(d-1) biseparable states in $\mathbb{C}^{d}\otimes ....\otimes \mathbb{C}^{d}$ which are entangled across every bipartition. From d-1 copies of such a state, it is possible to obtain a genuinely entangled state with some non-zero probability.   
\end{remark}

\section{Conclusion}\label{Conclusion}
Most of the existing protocols for distributing genuine multipartite correlations employ teleportation-based schemes where one party locally prepares a genuinely multipartite entangled state. In our schemes also, we can think about including a teleportation based step. However, on the other hand, our schemes can be thought of excluding any teleportation based step, introducing a conceptually simple \textcolor{black}{sequential extraction-and-fusion } protocol for the activation of genuine multipartite entanglement. Clearly, our  protocol does not always require any locally prepared genuinely entangled state. Furthermore, our protocols do not rely on the implementations of joint measurements on copies of quantum states in each step. It is also quite interesting to note that considering our proposed protocol in three qutrit systems, just two copies of biseparable states can lead to sharing of genuine multipartite correlations among the parties. Building upon these insights, we have extended our protocols to higher number of parties $n>3$, thereby broadening the scope of the protocols. Finally, we mention that the present protocols focus on producing pure genuinely entangled states which may exhibit genuine nonlocality. This is, in fact, stronger than activating genuine multipartite entanglement. We have also identified a class of biseparable states for which the success probability approaches unity in the asymptotic limit. It would be of significant interest to investigate how the success probability varies with an increasing number of copies across general classes of biseparable states. 
Developing systematic methods to address this broader and more complex problem would be interesting for further studies.\\


\textit{Acknowledgments}. SC and US acknowledges partial support from the Department of Science and Technology, Government of India, through the QuEST grant with Grant No. DST/ICPS/QUST/Theme-3/2019/120 via I-HUB QTF of IISER Pune, India.\\

\textit{Data availability}. No data were created or analysed in this study.
  \bibliography{ref}
\end{document}